\documentclass[12pt]{article}
\usepackage[utf8]{inputenc}

\usepackage{amssymb,amsmath,amsthm}
\usepackage{url}
\usepackage{xspace}
\usepackage{enumitem}

\usepackage[dvipsnames,usenames]{xcolor}
\usepackage{tikz}
\usetikzlibrary{decorations.pathreplacing}
\usetikzlibrary{calc}

\usepackage{todonotes}

\theoremstyle{plain}
\newtheorem{theorem}{Theorem}[section]
\newtheorem{lemma}[theorem]{Lemma}

\theoremstyle{definition} 

\newcommand{\Case}[2]{\smallskip\par{\it Case #1:\/ #2}}
\newcommand{\Subcase}[2]{\smallskip\par{\it Subcase #1:\/ #2}}

\newcounter{oq}

\newcommand{\refeq}[1]{(\ref{eq:#1})}

\newcommand{\of}[1]{\left( #1 \right)}

\newcommand{\E}{\exists}

\newcommand{\und}{\wedge}

\newcommand{\spa}{\mathit{sp}}

\newcommand{\classc}{\ensuremath{\mathcal C}\xspace}
\newcommand{\subgr}[1]{\ensuremath{\mathcal S(#1)}\xspace}

\title{Tight Bounds on the Asymptotic Descriptive Complexity of Subgraph Isomorphism}

\author{Oleg Verbitsky%
\thanks{Institut f\"ur Informatik,
  Humboldt-Universit\"at zu Berlin, Unter den Linden 6, D-10099 Berlin.
  Supported by DFG grant KO 1053/8--1.
  On leave from the IAPMM, Lviv, Ukraine.}
\and
Maksim Zhukovskii%
\thanks{Laboratory of Advanced Combinatorics and Network Applications, 
Moscow Institute of Physics and Technology, Moscow.
Supported by grant No.\ 16-31-60052 of Russian Foundation for Basic Research.}
}

\date{}

\begin{document}

\maketitle

\begin{abstract}
Let $v(F)$ denote the number of vertices in a fixed connected pattern graph $F$.
We show an infinite family of patterns $F$ such that the existence of a subgraph isomorphic
to $F$ is expressible by a first-order sentence of
quantifier depth $\frac23\,v(F)+1$, assuming that the host graph is sufficiently large
and connected. On the other hand, this is impossible for any $F$ with
using less than $\frac23\,v(F)-2$ first-order variables.
\end{abstract}

\section{Introduction}\label{s:intro}

We consider graph properties expressible in first-order logic
over the vocabulary consisting of two binary relations symbols,
$\sim$ for adjacency and $=$ for equality of vertices.
Let $F$ be a fixed pattern graph on the vertex set $\{1,\ldots,\ell\}$.
The Subgraph Isomorphism problem asks whether a given graph
contains a subgraph isomorphic to $F$, which can be expressed by
the first-order sentence
\begin{equation}
  \label{eq:defn}
\E x_1\ldots\E x_\ell\,\of{
\bigwedge_{i\ne j}x_i\ne x_j
\und
\bigwedge_{\{i,j\}\in E(F)}x_i\sim x_j
},  
\end{equation}
where $E(F)$ denotes the edge set of~$F$.

Consider the parameters $D(F)$ and $W(F)$ defined, respectively, as the minimum
quantifier depth and the minimum variable width of a first-order sentence
expressing Subgraph Isomorphism for the pattern graph $F$.
Note the relation $W(F)\le D(F)$, following from the general fact that
any first-order sentence of quantifier depth $d$ can be rewritten using 
at most $d$ variables. Since the sentence \refeq{defn} has quantifier depth $\ell$,
we have $D(F)\le\ell$. On the other hand, note that $K_\ell$, the complete graph
on $\ell$ vertices, contains $F$ as a subgraph, while the smaller complete graph
$K_{\ell-1}$ does not. Since $\ell$ first-order variables are necessary in order
to distinguish between $K_\ell$ and $K_{\ell-1}$, we have $W(F)\ge\ell$ and,
therefore, $W(F)=D(F)=\ell$. In other words, the existence of an $F$ subgraph
cannot in general be expressed more succinctly, with respect to the
quantifier depth or the variable width, than by the exhaustive description~\refeq{defn}.

Assume that the pattern graph $F$ is connected.
The time complexity of Subgraph Isomorphism will not be affected
if we restrict this problem to connected input graphs.
The same holds true for the descriptive complexity:
we still need $\ell$ variables to express Subgraph Isomorphism 
over connected graphs because the ``hard case'' was given by
connected graphs $K_\ell$ and $K_{\ell-1}$. Can it happen, however, 
that this pair is the only obstacle to expressing Subgraph Isomorphism
over connected graphs more succinctly? In fact, it is natural
(in full accordance with computational complexity theory!)
to consider the \emph{asymptotic} descriptive complexity of 
Subgraph Isomorphism over connected graphs. More precisely, let
$D'(F)$ be  the minimum quantifier depth of a first-order sentence
correctly deciding whether or not a graph $G$ contains an $F$ subgraph
over all \emph{sufficiently large connected} graphs $G$;
see Section \ref{ss:logic} for more details. Similarly, let $W'(F)$ denote
the asymptotic version of the width parameter $W(F)$.
The question addressed in this paper is whether the asymptotic
descriptive complexity of Subgraph Isomorphism can be much lower over connected graphs
than in general or, more formally, how much the asymptotic
parameters $D'(F)$ and $W'(F)$ can differ from their standard counterparts
$D(F)$ and~$W(F)$.

In our earlier paper \cite{VZh16}, we found an example of a pattern graph $F$
with $D'(F)\le\ell-3$ and observed, on the other hand, that
$W'(F)\ge\frac12\,\ell-\frac12$ for all $F$. It remained unknown whether
the difference between $W'(F)$ and $W(F)=\ell$ could be arbitrarily large.
We now show an infinite family of pattern graphs $F$ with $D'(F)\le c\cdot\ell$
for a constant factor $c<1$ and determine the minimum value of $c$,
for which such a bound is possible.
More precisely, we show that
\begin{equation}
  \label{eq:D}
D'(F)\le\frac23\,\ell+1  
\end{equation}
for infinitely many $F$, where $\ell$ always denotes the number of vertices in $F$.
We also prove that this upper bound is tight by accompanying it with a nearly
matching bound
\begin{equation}
  \label{eq:W}
W'(F)>\frac23\,\ell-2\text{ for every }F.
\end{equation}

A general reason why the existence of an $\ell$-vertex subgraph $F$ can be
defined in this setting with sharply less than $\ell$ first-order variables lies in the fact
that the logical truth changes if we restrict our realm to large connected graphs.
In particular, some special sentences about subgraph containment become
validities in this realm. 
As an example, consider the following statement:
\begin{description}
\item[($\Phi_\ell$)] 
A graph $G$ contains a subgraph on $\ell$ vertices isomorphic either to 
a path $P_\ell$ or to a star $K_{1,\ell-1}$.
\end{description}
The observation that, for each $\ell$, the statement $\Phi_\ell$ is true for all sufficiently large 
connected graphs~$G$ serves as a starting point
of the graph-theoretic work by Oporowski, Oxley, and Thomas~\cite{OporowskiOT93};
other examples of this kind can be found in Chapter 9.4 of Diestel's textbook~\cite{Diestel}.

Though $\Phi_\ell$ implies no impressive upper bounds for 
$W'(P_\ell)$ nor for $W'(K_{1,\ell-1})$,\footnote{%
These parameters are actually equal to $\ell-2$ and $\ell-1$, respectively;
see Section~\ref{s:exact}.}
this property underlies our analysis of an important hybrid pattern graph.
Specifically, the \emph{sparkler graph} $S_{q,p}$ is obtained by drawing an edge
between an end vertex of a path $P_p$ and the central vertex of a star $K_{1,q-1}$.
We determine the values of $D'(S_{q,p})$ and $W'(S_{q,p})$ up to a small additive constant.
Specifically,
\begin{eqnarray}
D'(S_{q,p})&\le&\max\left(p+2,\,\frac12\,p+q\right),\text{ while}\label{eq:Dspa}\\
W'(S_{q,p})&\ge&\max\left(p,\,\frac12\,p+q-\frac52\right).\label{eq:Wspa}
\end{eqnarray}

The right hand side of \refeq{Dspa} attains its minimum when $p=2q-4$, 
yielding the upper bound~\refeq{D}. Our proof of \refeq{Dspa} exploits
a structural property of connected $S_{q,p}$-free graphs closely related
to the aforementioned properties $\Phi_\ell$: The maximum vertex degree of such graphs
is bounded either from above by a constant or from below by an increasing
function of $\ell=q+p$ (see Lemma~\ref{lem:large-deg}).
Another important ingredient in our analysis is a dichotomy theorem by Pikhurko, Veith, and Verbitsky~\cite{PikhurkoVV06}
about succinct definability of an individual graph,
stated as Lemma \ref{lem:PVV} in Section~\ref{s:prel}.

The lower bound \refeq{Wspa} readily implies that $W'(S_{q,p})>\frac23\,\ell-2$
irrespectively of the parameters $q$ and $p$. This particular fact about the  sparkler graphs 
plays a key role in the proof of the general lower bound \refeq{W}.
Our argument for \refeq{W} actually reveals the structure of extremal patterns $F$
with $W'(F)\approx\frac23\,\ell$: Every such $F$ either has one of a few simple
combinatorial properties\footnote{%
namely those underlying Lemmas \ref{lem:sp-lower} and \ref{lem:sp}.}
or is a sparkler graph $S_{q,p}$ with $p\approx2q$.

We conclude the summary of our results with listing some reasons motivating 
investigation of the asymptotic descriptive complexity of subgraph containment problems
and, perhaps, also other first-order properties of graphs.

\paragraph{\it Relation to computational complexity.}
The model-check\-ing problem for a fixed first-order sentence $\Phi$ is solvable
in time $O(n^{W(\Phi)})$, where $n$ is the number of vertices in an input graph
and $W(\Phi)$ denotes the variable width of $\Phi$
(Immerman \cite{Immerman82}, Vardi \cite{Vardi95}). This implies the time bound $O(n^{W(F)})=O(n^\ell)$
for Subgraph Isomorphism for an $\ell$-vertex pattern graph $F$, which, of course,
corresponds to exhaustive search through all $\ell$-tuples of vertices in the input graph.
If $F$ is connected, then Subgraph Isomorphism efficiently reduces to its restriction
to connected inputs and the time bound $O(n^{W(F)})$ can be replaced with a potentially
better bound $O(n^{W'(F)})$. According to \refeq{D}, for some $F$ this results in time $O(n^{\frac23\,\ell+1})$,
which actually may look not so bad if compared to the general time bound  $O(n^{(\omega/3)\ell +2})$
established for Graph Isomorphism by Ne\v{s}et\v{r}il and Poljak \cite{NesetrilP85};
here $\omega$ is the exponent of fast matrix multiplication, whose value
is known~\cite{Gall14} to lie between $2$ and $2.373$.
However, all patterns $F$ with $W'(F)\approx\frac23\,\ell$ have a large tree part
and, for such graphs, the time bound $O(n^{\frac23\,\ell+1})$ cannot compete with other
algorithmic techniques for Subgraph Isomorphism such as the color-coding method
by Alon, Yuster and Zwick \cite{AlonYZ95}. 

Nevertheless, there is apparently no general reason why the time bounds based
on estimating the asymptotic descriptive complexity cannot be somewhat more
successful in some other situations. Consider, for example, the
\emph{Induced} Subgraph Isomorphism problem, whose computational complexity seems
very different from the not-necessarily-induced case.
Here it is unknown if the bound $O(n^{W'(F)})$ can, for some patterns $F$, be comparable
with the Ne\v{s}et\v{r}il-Poljak time bound. The last bound applies also to Induced Subgraph Isomorphism,
and no techniques achieving running time $O(n^{o(\ell)})$ with a sublinear exponent for infinitely
many patterns are known.
We discuss this topic in \cite{induced}; see also the concluding remarks in Section~\ref{s:concl}.

\paragraph{\it Encoding-independent computations and order-invari\-ant definitions.}
The above discussion shows that any first-order sentence in the vocabulary 
$\{\sim,=\}$ defining a graph property over sufficiently large graphs
can be considered a weak computational model for the corresponding
decision problem. The question of its efficiency is in the spirit of the 
eminent problem on the power of encoding-independent computations; see, e.g., \cite{GraedelG15}.
If we extend the vocabulary with the order relation $<$,
comparison of the two settings is interesting in the context of
\emph{order-invariant definitions}; see, e.g., \cite{Libkin04,Schweikardt13}.
Note also that the setting where arbitrary numerical relations are allowed
brings us in the field of circuit complexity; see \cite{Immerman-book,Libkin04}. 
In this context, the Subgraph Isomorphism problem has been studied in~\cite{LiRR14}.

\paragraph{\it Measurement of succinctness.}
A traditional question studied in finite model theory asks whether or not a graph
property of interest is expressible in a certain logical formalism.
If the expressibility is known, it is reasonable to ask how succinctly the property
can be expressed with respect to the length, the quantifier depth,
or the variable width of a defining sentence, and the questions we address
in this paper are exactly of this kind. We refer the reader to
Dawar \cite{Dawar98}, Grohe and Schweikardt \cite{GroheS05}, and
Tur{\'{a}}n \cite{Turan84}
for instances of the problems studied in this line of research.

\paragraph{Organization of the paper.}
Section \ref{s:prel} contains definitions and preliminary lemmas.
The bounds \refeq{Dspa} and \refeq{Wspa} for sparkler graphs are proved
in Section \ref{s:sparklers}. Our main result, the general bounds
\refeq{D}--\refeq{W}, is obtained in Section \ref{s:general}.
Furthermore, Section \ref{s:exact} is devoted to the particular case
of path and star graphs, in which we are able to determine
the values of $D'(F)$ and $W'(F)$ precisely. We conclude with
a discussion of further questions in Section~\ref{s:concl}.

\section{Preliminaries}\label{s:prel}

\subsection{Basic definitions}\label{ss:logic}

We consider the first-order language containing the adjacency and the equality relations.
We say that a first-order sentence $\Phi$ defines a class of graphs \classc
\emph{asymptotically over connected graphs} if there is an integer $N$
such that
$$
G\models\Phi\text{ iff }G\in\classc
$$
for all connected graphs $G$ with at least $N$ vertices.
The \emph{asymptotic logical depth} of \classc,
denoted by $D'(\classc)$, is the minimum quantifier depth (rank) of such~$\Phi$.

The \emph{variable width} of a first-order sentence $\Phi$ is the number of first-order variables
used to build $\Phi$; different occurrences of the same variable do not count.
The \emph{asymptotic logical width} of \classc,
denoted by $W'(\classc)$, is the minimum variable width of a $\Phi$
defining \classc asymptotically over connected graphs.
Note that
$
W'(\classc)\le D'(\classc)
$.

\subsection{Our toolkit}

Given two non-isomorphic graphs $G$ and $H$, let
$D(G,H)$ (resp.\ $W(G,H)$) denote the minimum quantifier depth (resp.\ variable width) of a sentence
distinguishing $G$ and $H$, that is, true on one of the graphs and false on the other.

\begin{lemma}\label{lem:DD}
\mbox{}

  \begin{enumerate}[label=\normalfont\bfseries\arabic*.]
\item
$D'(\classc)\le d$ if $D(G,H)\le d$ for all sufficiently large connected 
graphs $G\in\classc$ and $H\notin\classc$.
\item
$W'(\classc)\ge d$ if there are arbitrarily large connected graphs $G\in\classc$ and $H\notin\classc$
with $W(G,H)\ge d$.
  \end{enumerate}
\end{lemma}

\noindent
Part 1 of this lemma can easily be deduced from the fact that, over a given vocabulary,
there are only finitely many pairwise inequivalent first-order formulas of a fixed
quantifier depth. Part 2 follows directly from the definitions.

Lemma \ref{lem:DD} reduces estimating $D'(\classc)$ and $W'(\classc)$ to estimating,
respectively, the parameters $D(G,H)$ and $W(G,H)$ over large connected $G\in\classc$ and $H\notin\classc$.
For estimating $D(G,H)$ and $W(G,H)$ we have a very handy instrument.

The \emph{$k$-pebble Ehrenfeucht-Fra{\"\i}ss{\'e} game}
is played on two vertex-disjoint graphs $G$ and $H$.
This is a two-person game; the players are called \emph{Spoiler} and \emph{Duplicator}.
\emph{He} and \emph{she} have equal sets of $k$ pairwise different pebbles.
In each round, Spoiler takes a pebble and puts it on a vertex in $G$ or in $H$;
then Duplicator has to put her copy of this pebble on a vertex
of the other graph.
The pebbles can be reused and change their positions during the play.
Duplicator's objective is to ensure that the pebbling determines a partial
isomorphism between $G$ and $H$ after each round; when she fails, she immediately loses.
The proof of the following facts can be found in Immerman's textbook~\cite{Immerman-book}.

\begin{lemma}\label{lem:ehr}
\mbox{}

\begin{enumerate}[label=\normalfont\bfseries\arabic*.]
\item
 $D(G,H)$ is equal to the
minimum $k$ such that Spoiler has a winning strategy in the $k$-round $k$-pebble
game on $G$ and~$H$.
\item
 $W(G,H)$ is equal to the
minimum $k$ such that, for some $d$, Spoiler has a winning strategy in the $d$-round $k$-pebble
game on $G$ and~$H$.
\end{enumerate}
\end{lemma}

For bounding $W'(\classc)$ from below based on Lemma \ref{lem:DD}, we need
a supply of pairs of graphs $G,H$ with large $W(G,H)$.
The following simple construction works sometimes quite efficiently.

Let $u$, $v$, and $w$ be three pairwise different vertices of a graph.
We say that $w$ \emph{separates} $u$ and $v$ if $w$ is adjacent
to exactly one of these two vertices. We call $u$ and $v$ \emph{twins}
if these vertices are inseparable by any vertex $w$ or, equivalently,
if the transposition $(uv)$ is an automorphism of the graph.
We consider $u$ and $v$ twins also if $u=v$.
Being twins is an equivalence relation on the vertex set, and
every equivalence class, called \emph{twin class}, is a clique or an independent set,
that is, a \emph{homogeneous set}.
For a vertex $v$ in a graph $G$, let $G-v$ denote the graph obtained
by removing $v$ from~$G$.

\begin{lemma}\label{lem:rm1twin}
  Suppose that $G$ is a graph containing a twin class of size $t$.
Let $v$ be one of these $t$ twins. Then
$W(G,G-v)\ge t$.
\end{lemma}

\begin{proof}
  By part 2 of Lemma \ref{lem:ehr}, it suffices to notice that Duplicator
survives in the $(t-1)$-pebble game on $G$ and $G-v$, whatever the number of
rounds is played. Whenever Spoiler plays outside the twin class under consideration,
Duplicator just mirrors his moves in the other graph. Whenever Spoiler pebbles
one of the twins, Duplicator makes the same in the other graph;
the particular choice of a twin is immaterial. This is always possible
because the reduced twin class in $G-v$ contains as many vertices as the
number of pebbles.
\end{proof}

Let $\sigma(G)$ denote the maximum size of a twin class in a graph $G$.
Lemma \ref{lem:rm1twin} implies that we need at least $\sigma(G)$ first-order variables
in order to define $G$ as an individual graph, i.e., to distinguish $G$
from all non-isomorphic graphs. As a consequence, the quantifier depth
needed for this purpose cannot be smaller than $\sigma(G)$.
It turns out that this is the only reason why it can be large.
More specifically, the following result shows that every $n$-vertex graph $G$
either is definable with quantifier depth no larger than $\frac12\,n+\frac52$
or has $\sigma(G)>\frac12\,n+\frac12$, and in the latter case
the minimum quantifier depth of a sentence defining $G$ is very close to~$\sigma(G)$.

\begin{lemma}[{Pikhurko, Veith, and Verbitsky \cite[Theorem 4.1]{PikhurkoVV06}}]\label{lem:PVV}
If $A$ and $B$ are non-isomorphic graphs, then
$$
D(A,B)\le
\begin{cases}
\frac12\,v(A)+\frac52&\text{if }\sigma(A)\le\frac12\,v(A)+\frac12,\\
\sigma(A)+2&\text{if }\sigma(A)\ge\frac12\,v(A)+\frac12.
\end{cases}
$$
Moreover, if $\sigma(A)\ge\frac12\,v(A)+1$, then
$D(A,B)\le\sigma(A)+1$ whenever the largest twin class in $A$ is an inclusion-maximal homogeneous set.
\end{lemma}

\noindent
Here and throughout the paper, $v(G)$ denotes the number of vertices in a graph~$G$.

\subsection{Useful graphs}

The \emph{neighborhood} $N(v)$ of a vertex $v$ consists of
all vertices adjacent to $v$. The number of neighbors $|N(v)|$
is called the \emph{degree} of $v$ and denoted by $\deg v$.
The vertex of degree $v(G)-1$ (i.e., adjacent to all other vertices)
is called \emph{universal}.
The vertex of degree 1 is called \emph{pendant}.

We use the standard notation
$P_n$ for paths and $C_n$ for cycles on $n$ vertices.
Furthermore, $K_{t,s}$ denotes the complete bipartite graph whose
vertex classes have $t$ and $s$ vertices. In particular, $K_{1,n-1}$
is the star graph on $n$ vertices.

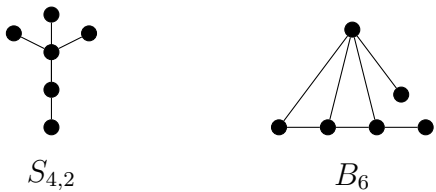
\begin{figure}
  \centering
\begin{tikzpicture}[every node/.style={circle,draw,inner sep=2pt,fill=black},scale=.5]

  \begin{scope}
    \path (0,0) node (a1) {}
      (0,1) node (a2) {} edge (a1)
      (0,2) node (a3) {} edge (a2)
      (0,3) node (a4) {} edge (a2)
   (-1,2.5) node (a5) {} edge (a3)
    (1,2.5) node (a6) {} edge (a3);
\node[draw=none,fill=none] at (0,-1.275) {$S_{4,2}$};
  \end{scope}

  \begin{scope}[xshift=80mm,scale=1.3]
    \path (-1.5,0) node (a1) {}
      (-.5,0) node (a2) {} edge (a1)
      (.5,0) node (a3) {} edge (a2)
      (1.5,0) node (a4) {} edge (a3)
   (0,2) node (a) {} edge (a1) edge (a2) edge (a3)
    (1,.67) node (y) {} edge (a);
\node[draw=none,fill=none] at (0,-1) {$B_6$};
  \end{scope}

\end{tikzpicture}
  \caption{Examples of useful graphs: a sparkler and a broken fan.} 
  \label{fig:special-graphs}
\end{figure}

The \emph{sparkler graph} $S_{q,p}$, playing an important role in the paper, 
is obtained from $K_{1,q-1}$ and $P_p$
by adding an edge between an end vertex of $P_p$ and the central vertex of~$K_{1,q-1}$.
The components $K_{1,q-1}$ and $P_p$ are referred to as the \emph{star} and \emph{tail parts}
of the sparkler graph respectively.
By the \emph{central vertex} of $S_{q,p}$ we mean the central vertex of
the star part.
We will almost always assume that $q\ge3$ and $p\ge2$, though the smaller parameters
still make sense representing $S_{q,1}\cong K_{1,q}$ and $S_{2,p}\cong P_{p+2}$.

The following graph will appear in Section \ref{ss:upper} for a technical purpose.
Recall that the $n$-vertex \emph{fan} graph is obtained
by adding a universal vertex $y$ to the path graph $P_{n-1}$. 
Let us ``break'' the edge between $y$ and one of the end vertices of $P_{n-1}$
by removing this edge from the graph and replacing it with an edge from
$y$ to a new vertex $y'$. This
results in an $(n+1)$-vertex graph that we call \emph{broken fan} and denote by~$B_{n+1}$;
see Fig.~\ref{fig:special-graphs}.

\subsection{First bounds}

Let \subgr F denote the class of graphs containing a subgraph
isomorphic to $F$. Simplifying the general notation introduced in Subsection \ref{ss:logic}, we write
$D'(F)=D'(\subgr F)$ and $W'(F)=W'(\subgr F)$.
We now state simple combinatorial bounds for these parameters that
were used already in~\cite{VZh16}.

Let $v_0v_1\ldots v_t$ be an induced path in a graph $G$.
We call it \emph{pendant} if $\deg v_0\ne2$, $\deg v_t=1$ and $\deg v_i=2$ for all $1\le i<t$.
Furthermore, let $S$ be an induced star $K_{1,s}$ in $G$
with the central vertex $v_0$. We call $S$ \emph{pendant}
if all its pendant vertices are pendant also in $G$,
and in $G$ there is no more than $s$ pendant vertices adjacent to $v_0$.
The definition ensures that a pendant path (or star) cannot be
contained in a larger pendant path (or star). 

Let $p(F)$ denote the maximum $t$ such that $F$ has a pendant path~$P_{t+1}$.
Similarly, let $s(F)$ denote the maximum $s$ such that $F$ has a pendant star~$K_{1,s}$.
As an example,
note that the sparkler graph $S_{q,p}$ has a pendant $P_{p+1}$
and a pendant~$K_{1,q-1}$ and, therefore, $p(S_{q,p})=p$ and $s(S_{q,p})=q-1$.
If $F$ has no pendant vertex, then we set $p(F)=0$ and $s(F)=0$.

\begin{lemma}\label{lem:sp-lower}
For every connected graph $F$ with $\ell$ vertices,
\begin{enumerate}[label=\normalfont\bfseries\arabic*.]
\item
$W'(F)\ge \ell-s(F)-1$, and
\item
$W'(F)\ge \ell-p(F)-1$. 
\end{enumerate}
\end{lemma}


\begin{proof}
\textbf{1.}
Since $s(K_{1,\ell-1})=\ell-1$, the bound is trivial for $F=K_{1,\ell-1}$,
and we assume that $F$ is not a star. By Part 2 of Lemma \ref{lem:DD},
it suffices to exhibit, for every sufficiently large $n$, a pair of graphs
$G_n$ and $H_n$ with at least $n$ vertices each such that $G_n$ contains $F$
as a subgraph, $H_n$ does not, and $W(G_n,H_n)\ge \ell-s(F)-1$.
For this purpose, consider $G_n$ consisting of the complete graph $K_{\ell-s(F)}$ with a pendant
star $K_{1,n}$ and $H_n$ obtained in the same way from the smaller complete
graph $K_{\ell-s(F)-1}$. Note that $H_n$ is obtainable from $G_n$ by removing
one of $\ell-s(F)-1$ twins. The lower bound for $W(G_n,H_n)$ follows by Lemma~\ref{lem:rm1twin}.

\textbf{2.}
Assume that $F$ is not a path graph, because otherwise the bound is trivial.
We proceed similarly to the first part, considering now $G_n$ consisting of $K_{\ell-s(F)}$ 
with a pendant path $P_{n}$ and $H_n$ obtained from $G_n$ by removing one of the twins.
\end{proof}

\section{Bounds for sparkler graphs}\label{s:sparklers}

We begin our analysis with a very instructive case of sparkler graphs,
for which we are able to determine the asymptotic depth and width parameters
with high precision.

\begin{theorem}\label{thm:sparkler}
Let $q\ge3$ and $p\ge2$.
Let $\ell=q+p$ denote the number of vertices of the sparkler graph~$S_{q,p}$.

\begin{enumerate}[label=\normalfont\bfseries\arabic*.]
\item 
$W'(S_{q,p})\ge\max\left(p,\,\ell-\frac12\,p-2-\frac12(p\bmod 2)\right)$.
\item 
$D'(S_{q,p})\le\max\left(p,\,\ell-\frac12\,p-2-\frac12(p\bmod 2)\right)+2$.
\end{enumerate}  
\end{theorem}

The proof of Theorem \ref{thm:sparkler} occupies the rest of this section.

\subsection{The lower bound}

Applied to a sparkler graph, Part 1 of Lemma \ref{lem:sp-lower} yields
$W'(S_{q,p})\ge\ell-s(S_{q,p})-1=\ell-q=p$. It remains to prove
the lower bound
\begin{equation}
  \label{eq:lower-qp}
 W'(S_{q,p})\ge\ell-\frac12\,p-2-\frac12(p\bmod 2)=q+\frac12\,p-2-\frac12(p\bmod 2).
\end{equation}
If $p=2$ or $p=3$, this bound immediately follows from Part 2 of Lemma~\ref{lem:sp-lower}.
We, therefore, assume that $p\ge4$.

In order to prove the bound \refeq{lower-qp},
we will construct two graphs $G_{a,b,n}$ and $H_{a,b,n}$ depending on integer parameters
$a$, $b$, and $n$. Adjusting appropriately the parameters 
$a$ and $b$, we will ensure that
\begin{enumerate}[label=(\alph*)]
\item\label{item:F} 
$G_{a,b,n}$ contains $S_{q,p}$ as a subgraph,
\item\label{item:noF}  
$H_{a,b,n}$ does not, and
\item\label{item:W}
$W(G_{a,b,n},H_{a,b,n})\ge q+\frac12\,p-2-\frac12(p\bmod 2)$
\end{enumerate}
for any choice of the parameter $n$.
Moreover, 
\begin{enumerate}[label=(\alph*)]\setcounter{enumi}{3}
\item\label{item:n} 
both $G_{a,b,n}$ and $H_{a,b,n}$ have more than $n$ vertices
\end{enumerate}
by construction.
Properties \ref{item:F}--\ref{item:n} will yield the bound \refeq{lower-qp} 
by Part 2 of Lemma \ref{lem:DD} as
the parameter $n$ can be taken arbitrarily large.

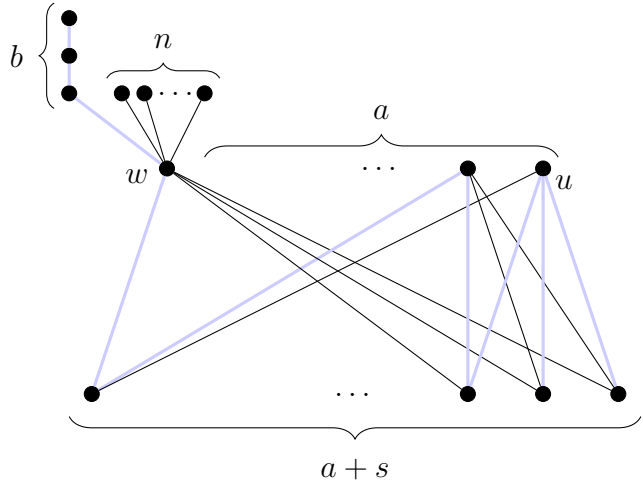
\begin{figure}
  \centering
\begin{tikzpicture}[every node/.style={circle,draw,inner sep=2pt,fill=black}]

\colorlet{mycol}{blue!20}

\path (-1,-3) node (c) {}
      (6,-3) node (d) {}
      (0,0) node (w) {} edge[very thick,mycol] (c) edge (d)
      (4,0) node (u) {} edge[very thick,mycol] (c) edge (d)
      (5,0) node (v) {} edge (c) edge[very thick,mycol] (d)
     (-.6,1) node (n1) {} edge (w)
      (-.3,1) node (n2) {} edge (w)
      (.5,1) node (n3) {} edge (w)
      (-1.3,1) node (b1) {} edge[very thick,mycol] (w)
     (-1.3,1.5) node (b2) {} edge[very thick,mycol] (b1)
     (-1.3,2) node (b3) {} edge[very thick,mycol] (b2)
      (5,-3) node (d1) {} edge (w) edge (u) edge[very thick,mycol] (v)
      (4,-3) node (d2) {} edge (w) edge[very thick,mycol] (u) edge[very thick,mycol] (v);

\node[draw=none,fill=none,below left] at ($(w)+(-.2,.1)$) {$w$};
\node[draw=none,fill=none,below right] at ($(v)+(.1,0)$) {$u$};

\draw[decorate,decoration={brace,amplitude=2mm,raise=2mm}] ($(b1)-(0,.2)$)--($(b3)+(0,.2)$);
\node[draw=none,fill=none] at (-2,1.5) {$b$};

\node[draw=none,fill=none] at (.15,1) {\ldots};
\draw[decorate,decoration={brace,amplitude=2mm,raise=2mm}] ($(n1)-(.2,0)$)--($(n3)+(.2,0)$);
\node[draw=none,fill=none] at (-.05,1.7) {$n$};

\node[draw=none,fill=none] at (2.85,0) {\ldots};
\draw[decorate,decoration={brace,amplitude=3mm,raise=1.5mm}] ($(w)+(.5,0)$)--($(v)+(.2,0)$);
\node[draw=none,fill=none] at (2.85,.75) {$a$};

\node[draw=none,fill=none] at (2.5,-3) {\ldots};
\draw[decorate,decoration={brace,mirror,amplitude=3mm,raise=3mm}] ($(c)+(-.3,0)$)--($(d)+(.3,0)$);
\node[draw=none,fill=none] at (2.5,-4) {$a+s$};

\end{tikzpicture}
\caption{Construction of $G_{a,b,n}$ by gluing $K_{a+1,a+s}$ and $S_{n+1,b}$
at the central vertex $w$ of the latter. Like the shown subgraph $S_{3,7}$ (with central vertex~$u$),
the graph $G_{a,b,n}$ hosts a copy of $S_{s+1,2a+b}$ as a subgraph.}
\label{fig:G(F)}
\end{figure}

Denote $s=q-1$.
The construction of $G_{a,b,n}$ is illustrated in Fig.~\ref{fig:G(F)}.
This graph is composed by two subgraphs sharing one common vertex.
One of them is the complete bipartite graph $K_{a+1,a+s}$, 
that will be referred to as the \emph{$K$-component}.
The other is the sparkler graph $S_{n+1,b}$,
that will be referred to as the \emph{$S$-component}. 
The smaller part of the $K$-component includes the central vertex $w$ of the $S$-component.

Note that the larger part of the $K$-component is an independent set consisting of $a+s$ twins. 
The graph $H_{a,b,n}$ is obtained from $G_{a,b,n}$ by removing one vertex from this twin class.
Lemma \ref{lem:rm1twin} implies that
\begin{equation}\label{eq:aa-1}
W(G_{a,b,n},H_{a,b,n})\ge a+s.  
\end{equation}

Now, we fix the parameters $a$ and $b$ such that Conditions \ref{item:F}--\ref{item:W}
are fulfilled. Specifically, we set
$$
b=2+(p\bmod 2)\text{ and }a=\frac{p-b}2.
$$
Note that $b\in\{2,3\}$. 
The value of $a$ is integer due to the choice of $b$.
Our assumption that $p\ge4$ ensures that $a\ge1$.
Note that Condition \ref{item:W} now readily follows from~\refeq{aa-1}.

As easily seen, the graph $G_{a,b,n}$ contains a copy of $S_{q,p}$.
Indeed, we can put the central vertex at any vertex $u\ne w$ in the smaller
part the $K$-component. The larger part the $K$-component is large enough to accommodate
$s$ vertices adjacent to $u$ and, moreover, there still remains enough space to draw
a path from $u$ to $w$ via $2a-1$ intermediate vertices.
Prolonging it along the tail part of the $S$-component, 
we obtain a path emanating from $u$ and passing through $2a+b=p$
further vertices.

This path is clearly destroyed by removing a vertex from the larger part of the $K$-component.
It remains to prove that no copy of $S_{q,p}$ can be found in $H_{a,b,n}$ in any other way.
Assume, to the contrary, that $H_{a,b,n}$ contains an isomorphic copy $S'$ of $S_{q,p}$
with central vertex $u'$. Consider several cases.

\begin{itemize}[topsep=1mm,leftmargin=0mm,labelsep=1mm,itemindent=4mm,listparindent=4mm,itemsep=1.5mm]
\item
\textit{$u'\ne w$ belongs to the smaller part of the $K$-component of $H_{a,b,n}$.}
After locating the star part of $S'$ in the larger part of the $K$-component,
this part contains only $a-1$ unoccupied vertices. It follows that
the longest path that can be drawn from $u'$ has length $2a-2+b=p-2$.
\item
\textit{$u'=w$.}
Now the star part of $S'$ can be located among the $n$ pendant vertices of $H_{a,b,n}$
adjacent to $w$. However, by the assumption that $p\ge4$, the tail part of $S'$ does not fit
into the tail part of the $S$-component of $H_{a,b,n}$.
The longest path from $w$ in the $K$-component has length $2a+1<2a+b=p$.
\item
\textit{$u'$ belongs to the larger part of the $K$-component of $H_{a,b,n}$.}
Suppose that $s\le a+1$ because otherwise the star part of $S'$ does not fit
into the smaller part of the $K$-component. 
After locating the star part of $S'$ in this part,
the longest path that can be drawn in the $K$-component from $u'$ has length 
$2(a+1-s)<2a<p$ if it terminates in the larger part of the $K$-component.
Otherwise such a path has length at most $2(a+1-s)-1$ and, arriving at $w$,
can be prolonged in the $S$-component of $H_{a,b,n}$ to a path
of total length at most $2(a+1-s)-1+b=p+1-2s<p$.
\end{itemize}

We get a contradiction in each of the cases, which completes our analysis.
The proof of Part 1 of Theorem \ref{thm:sparkler} is complete.

\subsection{The upper bound}\label{ss:upper}

We now turn to proving the upper bound of Theorem \ref{thm:sparkler}.
We begin with a few simple properties of graphs without $S_{q,p}$ subgraphs.
The following lemma generalizes a property of $S_{4,4}$-free graphs observed in~\cite{VZh16}.

\begin{lemma}\label{lem:large-deg}
Suppose that a connected graph $H$ does not contain a subgraph $S_{q,p}$. 
Then either $\Delta(H)<q$ or $\Delta(H)\ge(v(H)/3)^{1/(2qp)}$.
\end{lemma}

\begin{proof}
Assume that $\Delta(H)\ge q$, that is, $H$ contains a subgraph $K_{1,q}$.
Then $H$ cannot contain any subgraph $P_{2qp}$ because, together with $K_{1,q}$,
it would give an $S_{q,p}$ subgraph. Consider an arbitrary spanning tree $T$ in $H$
and denote its maximum vertex degree $d$ and its radius by $r$.
Note that 
$$
v(H)=v(T)\le 1+d+d(d-1)+\ldots+d(d-1)^{r-1}\le3d^r.
$$
Since $r\le 2qp$, we have $v(H)\le3d^{2qp}$. It follows that $\Delta(H)\ge d\ge(v(H)/3)^{1/(2qp)}$.
\end{proof}


\begin{lemma}\label{lem:noCB}
Suppose that a connected graph $H$ with at least $3(q+p)^{2qp}$ 
vertices does not contain a subgraph $S_{q,p}$. 
If $\Delta(H)\ge q$, then $H$ contains neither the cycle $C_{p+1}$ nor the broken fan $B_{p+2}$
as subgraphs.
\end{lemma}

\begin{proof}
By Lemma~\ref{lem:large-deg}, $H$ has a vertex $z$ of degree at least $(v(H)/3)^{1/(2qp)}\ge q+p$.
Assume, to the contrary, that $H$ contains $C_{p+1}$ or $B_{p+2}$.
We show that such a subgraph could be used
for building an $S_{q,p}$ subgraph with central 
vertex $z$ (which is impossible
by assumption). Indeed, if $z$ belongs to a cycle of length $p+1$,
then the cycle provides the tail part of $S_{q,p}$, and $z$ has sufficiently many
neighbors to build also the star part of $S_{q,p}$. If there is a $C_{p+1}$ subgraph
not containing $z$, then we can even find $S_{q+1,p}$ in $H$ by considering
an arbitrary path from $z$ to this cycle.

If $H$ contains a $B_{p+2}$ subgraph, denote its vertices
by $y,y',\allowbreak y_1,y_2,\ldots,y_p$ where $y_1,y_2,\ldots,y_p$ appear (in this order)
along the path part of $B_{p+2}$, $y$ is adjacent to all of these vertices but $y_p$,
and $y'$ is adjacent to $y$.
If $z=y_i$ with $1<i<p-1$, then the tail part of $S_{q,p}$ in $H$
is formed by the path $zy_{i-1}\ldots y_1yy_{i+1}\ldots y_p$.
The case of $z=y_1$ is similar. If $z=y_{p-1}$, then the tail part of $S_{q,p}$
is $zy_{p-2}\ldots y_1yy'$. If $z=y$, then the tail is $zy_1\ldots y_p$.
If $z$ is $y'$ or $y_p$ or does not belong to the $B_{p+2}$ subgraph,
then we can find a tail even longer than needed.
\end{proof}

We now restate Part 2 of Theorem \ref{thm:sparkler} as
\begin{equation}
  \label{eq:upper-qp}
D'(S_{q,p})\le\max\left(p+2,\,q+\frac12\,p-\frac12(p\bmod 2)\right).  
\end{equation}
Let $G$ and $H$ be two graphs such that $G$ contains $S_{q,p}$ as a subgraph and $H$
does not. Moreover, assume that $H$ is connected and sufficiently large;
specifically, $v(H)\ge3(q+p)^{2qp}$. 
Using Part 1 of Lemma \ref{lem:DD}, we have to design a strategy allowing
Spoiler to win the Ehrenfeucht game on $G$ and $H$ so fast that the number of moves
corresponds to the bound~\refeq{upper-qp}.

Let $S$ be a subgraph of $G$ isomorphic to $S_{q,p}$.
Denote the central vertex of $S$ by $x$, and let Spoiler pebble $x$
in the first round. Let $y$ denote the vertex pebbled in response
by Duplicator in $H$. If $\deg y<q$, then Spoiler wins in the next $q$ rounds
by pebbling $q$ neighbors of $x$. We, therefore, suppose that
\begin{equation}
  \label{eq:degy}
 \deg y\ge q.
\end{equation}
This implies that $\Delta(H)\ge q$ and makes Lemma \ref{lem:noCB} applicable to~$H$.

Let $xx_1x_2\ldots x_p$ be the tail path of $S$.
If $x$ and $x_p$ are adjacent, this yields a $C_{p+1}$ subgraph in $G$,
and Spoiler wins by pebbling it due to Lemma \ref{lem:noCB}.
We, therefore, suppose that $x$ and $x_p$ are not adjacent.

Denote the number of neighbors of $x$ among $x_1,x_2,\ldots,x_{p-1}$ by~$r$.
We also suppose that Spoiler cannot win in $p+1$ moves just by pebbling the tail part of $S$,
that is, Duplicator manages to keep a partial isomorphism between
$G$ and $H$ in this case by pebbling some path $yy_1y_2\ldots y_p$ in $H$.
Note that $y$ cannot be adjacent to $y_p$ and must have exactly $r$ neighbors among $y_1,y_2,\ldots,y_{p-1}$.
Since $y$ cannot be the central vertex of any $S_{q,p}$ subgraph in $H$,
this implies that
\begin{equation}
  \label{eq:yrq}
\deg y < r+q-1.  
\end{equation}

By definition, we have $r\le p-1$.
In fact, this inequality is strict under the additional condition that $q\ge p$. 
Indeed, the equality $r=p-1$ means that
$y$ is adjacent to all $y_1,y_2,\ldots,y_{p-1}$.
In the case that $q\ge p$, the inequality \refeq{degy} shows that $y$ also has
a neighbor $y'$ different from any $y_i$, $1\le i\le p$. This yields a $B_{p+2}$
subgraph in $H$, a contradiction.

Therefore, the inequality \refeq{yrq} implies that
\begin{equation}
  \label{eq:degyless}
 \deg y\le q+p-3
\end{equation}
and, moreover,
\begin{equation}
  \label{eq:degyless+}
 \deg y\le q+p-4\text{ if }q\ge p.
\end{equation}

Let $G'=G[N(x)]$ and $H'=H[N(y)]$, where $G[X]$ denotes the induced subgraph of $G$
on a set of vertices $X$. Thus,
$G'$ has $\deg x$ vertices and $H'$ has $\deg y$ vertices.
Whenever Spoiler moves in $G'$, Duplicator is forced to respond in $H'$ and vice versa
because otherwise she loses immediately. For the graphs $G'$ and $H'$ we
now apply Lemma~\ref{lem:PVV}.

Set $A=H'$ and $B=G'$, and consider several cases. 
Assume first that $\sigma(H')\le\frac12\,v(H')+\frac12$.
By Lemma \ref{lem:PVV},
$$
D(H',G')\le\frac12\,v(H')+\frac52.
$$
It follows by \refeq{degyless} that $D(H',G')\le\frac12\,q+\frac12\,p+1$,
that is, Spoiler wins the game on $G'$ and $H'$ making no more than $\frac12\,q+\frac12\,p+1$ moves.
This allows him to win the game on $G$ and $H$ making no more than $\frac12\,q+\frac12\,p+2$ moves,
implying
\begin{equation}
  \label{eq:DGHqp}
D(G,H)\le\frac12\,q+\frac12\,p+2.
\end{equation}
This yields the bound \refeq{upper-qp} if
$\frac12\,q+\frac12\,p+2\le p+2$. The last inequality is false
exactly when $q>p$. But then we can use \refeq{degyless+} and, similarly to \refeq{DGHqp},
obtain a little bit better bound, namely
$$
D(G,H)\le\frac12\,q+\frac12\,p+\frac32.  
$$
Again, this implies the bound \refeq{upper-qp} if $\frac12\,q+\frac12\,p+\frac32\le p+2$.
The last inequality is false exactly when $q\ge p+2$. In this case
the bound \refeq{upper-qp} is also true because
$\frac12\,q+\frac12\,p+\frac32\le q+\frac12\,p-\frac12(p\bmod 2)$.
The last inequality is true because otherwise we would have
$q<3+(p\bmod 2)\le4$, contradicting the condition $q\ge p+2\ge4$.

Assume now that $\sigma(H')\ge\frac12\,v(H')+1$. By Lemma \ref{lem:PVV}
we now have
$$
D(H',G')\le\sigma(H')+2.
$$
Our further analysis depends on whether the largest twin class $T$ in $H'$
is a clique or an independent set.

Assume the former. Note that $|T|\le p-1$.
Indeed, the clique $T$ cannot contain more vertices for else
$T\cup\{y\}$ would be a clique of size at least $p+1$ in $H$,
and $H$ would hence contain a $C_{p+1}$ subgraph, contradicting Lemma \ref{lem:noCB}.
Thus, $\sigma(H')\le p-1$. 
Therefore, $D(H',G')\le p+1$ and $D(H,G)\le p+2$, and the bound \refeq{upper-qp} follows.

Assume now that the largest twin class in $H'$ is an independent set.
By the assumption we made above,
$H$ contains a path $y_1y_2\ldots y_{p-1}$ such that at most $q-2$ vertices of $N(y)$
do not belong to this path.
It follows that 
\begin{equation}
  \label{eq:sigma}
\sigma(H')\le q-2+\alpha(P_{p-1}).  
\end{equation}
Here, $\alpha(K)$ denotes the independence number of a graph $K$; that is,
$$
\alpha(P_{p-1})=\frac12\,p-\frac12(p\bmod 2).
$$
If the inequality \refeq{sigma} is strict, we obtain
$D(H',G')\le\sigma(H')+2\le q-1+\alpha(P_{p-1})$
and 
$$
D(H,G)\le q+\frac12\,p-\frac12(p\bmod 2).
$$
It remains to consider the case that $\sigma(H')=q-2+\alpha(P_{p-1})$.
This equality implies that the largest twin class in $H'$ is an inclusion-maximal independent set.
Under this condition, Lemma \ref{lem:PVV} gives us the better relation $D(H',G')\le\sigma(H')+1$,
and the same bound for $D(H,G)$ stays true.

The proof of Theorem \ref{thm:sparkler} is complete.






\section{General bounds}\label{s:general}

We are now prepared to prove our main result.

\begin{theorem}\label{thm:gen}
\mbox{}

\begin{enumerate}[label=\normalfont\bfseries\arabic*.]
\item
$D'(F)\le\frac23\,v(F)+1$ for infinitely many connected $F$.
\item
$W'(F)>\frac23\,v(F)-2$ for every connected $F$.
\end{enumerate}
\end{theorem}

We prove Theorem \ref{thm:gen} in the rest of this section.
Part 1 readily follows from Part 2 of Theorem \ref{thm:sparkler} by
setting $q=t+2$ and $p=2t-2$.

In order to prove the lower bound, we introduce the following notion.
Let $v$ be a vertex of a connected graph $F$. By a \emph{$v$-branch} of $F$ we mean a subgraph
of $F$ induced by the vertex set of a connected component of $F\setminus v$
in union with the vertex $v$ itself. 
Let $S$ be a $v$-branch of $F$. We call $S$ a \emph{pendant sparkler subgraph $S_{q,p}$} of $F$ if
\begin{itemize}
\item 
$S$ is isomorphic to $S_{q,p}$,
\item 
$v$ is the end vertex of the tail part of $S$, and
\item 
$v$ has degree at least 3 or exactly 1 in $F$.
\end{itemize}
Note that $\deg v=1$ means that $F=S$, that is, $F$ is a sparkler graph itself.
As usually, we suppose that $q\ge3$. However, we allow $p=1$ ($S$ is a star $K_{1,q}$ attached at one
of its leaves to the rest of the graph).
In addition, we even allow $p=0$ ($S$ is a pendant star $K_{1,q-1}$ attached
at the central vertex). In the case that $F$ has at least one pendant sparkler subgraph,
we write $\spa(F)$ to denote the maximum $p$ such that $F$ has a pendant $S_{q,p}$ for some $q\ge3$.

\begin{lemma}\label{lem:sp}
Let $F$ be a connected graph with $\ell$ vertices. If $F$ has a pendant sparkler subgraph, then
$W'(F)\ge\ell-\spa(F)-3$.  
\end{lemma}

\begin{proof}
We suppose that $F$ is not a sparkler graph because otherwise
the required bound readily follows from Part 2 of Lemma \ref{lem:sp-lower}.

  Denote $p=\spa(F)$. Let $n\ge3$ be an integer parameter.
Let $G$ be obtained from $K_\ell$ and $S_{n,p+1}$ by gluing these
graphs at the end vertex of the tail of $S_{n,p+1}$. Thus, $G$
can be seen as a pendant $S_{n,p+1}$ attached to $K_\ell$. 
Let $H$ be obtained similarly from $K_{\ell-p-3}$ and $S_{n,p+1}$.

Note that $G$ contains $F$ as a subgraph just because the $K_\ell$ part of $G$
is large enough to contain any graph with $\ell$ vertices.
On the other hand, $H$ does not contains an $F$ subgraph.
Indeed, since the tail of the $S_{n,p+1}$ part of $H$ is 1 larger than $p$,
it can host at most a path of length $p+2$ (which is useful if $F$ has a pendant path). 
However, the remaining part of $F$ has at least $\ell-p-2$ vertices and
does not fit into the $K_{\ell-p-3}$ part of~$H$.

Now, note that the $K_\ell$ part of $G$ with the end vertex of $S_{n,p+1}$ excluded
is a twin class in $G$. Moreover, $H$ is obtained from $G$ by removing
$p+3$ twins from this class. Similarly to Lemma \ref{lem:rm1twin},
we have $W(G,H)\ge\ell-p-3$ (which can also be formally deduced from Lemma \ref{lem:rm1twin}
by repeatedly applying it $p+3$ times and using transitivity of the logical equivalence
in the $(\ell-p-4)$-variable logic). Since the parameter $n$ can be chosen arbitrarily large,
we conclude that $W'(F)\ge\ell-p-3$.
\end{proof}

We are now ready to prove Part 2 of Theorem \ref{thm:gen}. 
Assume that $W'(F)<\frac23\,\ell-1$ as otherwise we are done. By Lemma \ref{lem:sp-lower}, this implies that
\begin{eqnarray}
 s(F)&>&\frac13\,\ell\text{\quad and}\label{eq:sell}\\
 p(F)&>&\frac13\,\ell.\label{eq:pell}
\end{eqnarray}
The former estimate implies that $F$ has a pendant sparkler subgraph.
Lemma \ref{lem:sp} is, therefore, applicable. It implies that
\begin{equation}
  \label{eq:spell}
 \spa(F)> \frac13\,\ell-2.
\end{equation}
From \refeq{sell}--\refeq{spell} we conclude that $F$ contains a pendant star $F_1$,
a pendant path $F_2$, and a pendant sparkler $F_3$ such that each of these subgraphs
has more than $\frac13\,\ell+1$ vertices. 
If $F_1$ and $F_3$ are vertex-disjoint or share only one vertex,
which must be the central vertex of $F_1$ and the end vertex of the tail part of $F_3$,
these subgraphs together occupy more than $\frac23\,\ell+1$ vertices,
and then there remains not enough space for the pendant path $F_2$.
Therefore, $F_1$ and $F_3$ share at least two vertices, which
is possible only if $F_1$ coincides with the star part of $F_3$.
In this case, $F_3$ contains more than $\frac23\,\ell-1$ vertices.
It follows that $F_3$ and $F_2$ cannot be disjoint, for else they together
would occupy more than $\ell$ vertices.
Therefore, $F_3$ and $F_2$ overlap, which is possible
only if $F_2$ is the tail part of $F_3$, that is, $F=F_3$ is a sparkler graph.

Now, let $F=S_{q,p}$. By Part 1 of Theorem \ref{thm:sparkler}, 
$$
W'(S_{q,p})\ge\max\left(p,\,\ell-\frac12\,p-\frac52\right).
$$
As easily seen, the minimum of the right hand side over real $p\in(0,\ell)$
is attained at $p=\frac23\,\ell-\frac53$. It follows that $W'(F)>\frac23\,\ell-2$,
completing the proof.

\section{Exact values for stars and paths}\label{s:exact}

Theorem \ref{thm:sparkler} shows that, for sparkler graphs, the parameters
$W'(S_{q,p})$ and $D'(S_{q,p})$ take values in a segment of three
consecutive integers (even two if $p$ is odd). We now determine the
exact values of $W'(F)$ and $D'(F)$ for the two kinds of ``marginal'' sparklers,
namely for the star graphs $K_{1,\ell-1}$ and the path graphs~$P_\ell$.

We begin with the stars.
Note that $D'(K_{1,2})=1$ because every connected graph with at least 3 vertices contains $K_{1,2}=P_3$
as a subgraph. 

\begin{theorem}\label{thm:K_1s-v}
Let $s\geq 3$. Then
\begin{enumerate}[label=\normalfont\bfseries\arabic*.]
\item
$D'(K_{1,s})=s+1$;
\item
$W'(K_{1,s})=s$.
\end{enumerate}
\end{theorem}

\begin{proof}
Let $M_{s,t}$ denote the graph obtained
by subdividing each edge of $K_{1,s}$ into $P_{t+1}$; thus $v(M_{s,t})=st+1$.

\textbf{1.}
Consider $G=M_{s,t}$ and $H=M_{s-1,t}$. Obviously, $G$ contains $K_{1,s}$,
while $H$ does not. We now describe a strategy for Duplicator
allowing her to win the $s$-round game on $G$ and $H$, irrespectively of how large $t$ is.

Note that $G$ is obtained from $H$ by adding another, $s$-th
branch $P_{t}$. During the first $s-1$ rounds, at least one
branch of $G$ stays free of pebbles, and this allows Duplicator
to use a mirror strategy according to an isomorphism between
$H$ and a subgraph of $G$.
Note that, according to this strategy,
Duplicator pebbles the central vertex of one graph exactly when
Spoiler pebbles the central vertex in the other graph.
To describe the last, $s$-th round
of the game, we consider two cases.

\Case1{The central vertex was pebbled in the first $s-1$ rounds.}
Then there is at least one completely free branch in $H$
and at least two such branches in $G$. Duplicator can use
the mirror strategy in the $s$-th round too.

\Case2{The central vertex was not pebbled in the first $s-1$ rounds.}
This case is different from the previous one only when
each of the $s-1$ branches of $H$ contains a pebbled vertex,
and exactly one branch remains free in $G$.
Assuming this, we consider two subcases.

\Subcase{2-a}{Spoiler does not move in the free branch of $G$.}
The mirroring strategy is still available to Duplicator.
Note that this case includes the possibility that Spoiler pebbles the central vertex
in one of the graphs.

\Subcase{2-b}{Spoiler pebbles a vertex $u$ in the free branch of $G$}.
In this case, the central vertices of $G$ and $H$ are not pebbled.
This implies that $u$ is not adjacent to any other vertex pebbled in $G$.
If $t\ge4$, Duplicator is able to find a vertex $u'$ in $H$
not adjacent to the vertices previously pebbled in this graph.

\medskip

\textbf{2.}
\textit{The lower bound $W'(K_{1,s})\ge s$.}
Like in the proof of Part 1, take $G=M_{s,t}$, $H=M_{s-1,t}$,
and notice that Duplicator wins the $(s-1)$-pebble game on $G$ and $H$
whatever the number of rounds.

\textit{The upper bound $W'(K_{1,s})\le s$.}
Using $s$ variables, we have to write down a sentence expressing the existence
of a $K_{1,s}$ subgraph in a sufficiently large connected graph $G$.
Consider $\Phi_s$ saying that
$$
\E x_1\ldots\E x_s\of{
\bigwedge_{1\le i<j\le s}x_i\ne x_j
\und
\bigwedge_{1\le i\le s}\E x_i(\bigwedge_{j\ne i}x_i\sim x_j)
}.
$$
This sentence is obviously true on every graph containing $K_{1,s}$.
Suppose that a graph $H$ is connected, has more than $2s$ vertices,
and does not contain $K_{1,s}$. Let us prove that $\Phi_s$ is false on $H$.
Assuming the opposite, that is, $H\models\Phi_s$, consider
vertices $v_1,\ldots,v_s$ of $H$ whose existence is claimed by $\Phi_s$.
Then, for each $i\le s$, there is a vertex $u_i$ adjacent to all $v_1,\ldots,v_s$
but $v_i$. All $u_i$ are pairwise distinct
(if $u_i=u_j$, this vertex would have degree $s$). It follows that
$\deg u_i=s-1$ for each $i\le s$. Since $H$ is connected and has more than $2s$
vertices, at least one of the vertices in $\{v_1,\ldots,v_s,u_1,\ldots,u_s\}$
must have a neighbor outside this set. Such a vertex must have degree at least $s$,
a contradiction.
\end{proof}

We now turn to the path graphs.

\begin{theorem}\label{thm:P_n}
\mbox{}

  \begin{enumerate}[label=\normalfont\bfseries\arabic*.]
\item
$D'({P_\ell})=\ell-1$ for all $\ell\geq 4$.
\item
$W'({P_\ell})=\ell-2$ for all $\ell\geq 3$.
\end{enumerate}
\end{theorem}

\begin{proof}
As we already mentioned, $D'(P_3)=W'(P_3)=1$ by trivial reasons. 
The equalities $D'(P_4)=3$ and $W'(P_4)=2$ are proved in \cite{VZh16}.
We, therefore, suppose that $\ell\ge5$.

\medskip

\textit{The upper bound $D'({P_\ell})\le\ell-1$.}
Let $G$ and $H$ be two graphs, each with at least $\ell$ vertices.
Suppose that $H$ contains no $P_\ell$ as a subgraph, while $G$ contains
a path $a_1a_2\ldots a_\ell$. Let $m$ be the largest number such
the vertices $a_1,a_2,\ldots,a_m$ form a clique. We split our analysis
into three cases.

\Case 1{$m\ge\ell-1$.}
Spoiler pebbles the clique $a_1,a_2,\ldots,a_{\ell-1}$ and wins because
there is no $K_{\ell-1}$ in $H$ (having a $K_{\ell-1}$ and at least one more
vertex, $H$ would contain $P_\ell$ by connectedness).

\Case 2{$3\le m\le\ell-2$.}
By the definition of $m$, the vertex $a_{m+1}$ is not adjacent
to at least one of the preceding vertices. Without loss of generality,
we can suppose that $a_1$ and $a_{m+1}$ are not adjacent
(see Fig.~\ref{fig:proof:P_n} where the proof is illustrated
in the case of~$\ell=5$; in Case 2 we then have $m=3$).

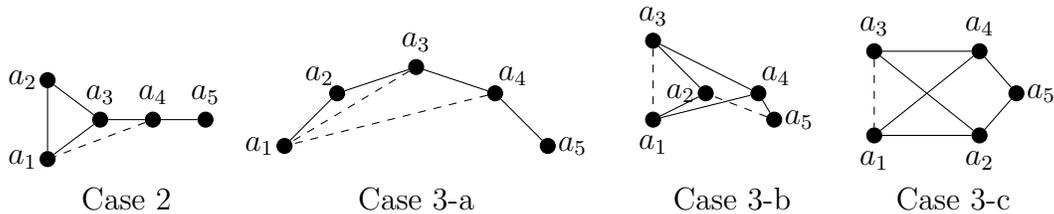
\begin{figure}
  \centering
\begin{tikzpicture}[every node/.style={circle,draw,inner sep=2pt,fill=black},scale=.7]
  \begin{scope}
\path (0,-.75) node (a1) {}
       (0,.75) node (a2) {} edge (a1)
       (1,0) node (a3) {} edge (a1) edge (a2)
       (2,0) node (a4) {} edge (a3) edge[dashed] (a1)
       (3,0) node (a5) {} edge (a4);
\node[draw=none,fill=none,left] at (a1) {$a_1$};
\node[draw=none,fill=none,left] at (a2) {$a_2$};
\node[draw=none,fill=none,above] at (a3) {$a_3$};
\node[draw=none,fill=none,above] at (a4) {$a_4$};
\node[draw=none,fill=none,above] at (a5) {$a_5$};
\node[draw=none,fill=none] at (1.5,-1.5) {Case 2};
  \end{scope}
  \begin{scope}[xshift=45mm,yshift=-5mm]
\path (0,0) node (a1) {}
       (1,1) node (a2) {} edge (a1)
       (2.5,1.5) node (a3) {} edge (a2) edge[dashed] (a1)
       (4,1) node (a4) {} edge (a3) edge[dashed] (a1)
       (5,0) node (a5) {} edge (a4);
\node[draw=none,fill=none,left] at (a1) {$a_1$};
\node[draw=none,fill=none,above left] at (a2) {$a_2$};
\node[draw=none,fill=none,above] at (a3) {$a_3$};
\node[draw=none,fill=none,above right] at (a4) {$a_4$};
\node[draw=none,fill=none,right] at (a5) {$a_5$};
\node[draw=none,fill=none] at (2.5,-1) {Case 3-a};
  \end{scope}
  \begin{scope}[xshift=115mm,yshift=5mm]
\path (0,-.5) node (a1) {}
       (1,0) node (a2) {} edge (a1)
       (0,1) node (a3) {} edge (a2) edge[dashed] (a1)
       (2,0) node (a4) {} edge (a3) edge (a1)
       (2.3,-.5) node (a5) {} edge (a4) edge[dashed] (a2);
\node[draw=none,fill=none,below] at (a1) {$a_1$};
\node[draw=none,fill=none,left] at (a2) {$a_2$};
\node[draw=none,fill=none,above] at (a3) {$a_3$};
\node[draw=none,fill=none,above right] at (a4) {$a_4$};
\node[draw=none,fill=none,right] at (a5) {$a_5$};
\node[draw=none,fill=none] at (1.5,-2) {Case 3-b};
  \end{scope}
  \begin{scope}[xshift=162mm,yshift=5mm]
\path (-.5,-.8) node (a1) {}
       (1.5,-.8) node (a2) {} edge (a1)
       (-.5,.8) node (a3) {} edge (a2) edge[dashed] (a1)
       (1.5,.8) node (a4) {} edge (a3) edge (a1)
       (2.2,0) node (a5) {} edge (a4) edge (a2);
\node[draw=none,fill=none,below] at (a1) {$a_1$};
\node[draw=none,fill=none,below] at (a2) {$a_2$};
\node[draw=none,fill=none,above] at (a3) {$a_3$};
\node[draw=none,fill=none,above] at (a4) {$a_4$};
\node[draw=none,fill=none,right] at (a5) {$a_5$};
\node[draw=none,fill=none] at (1,-2) {Case 3-c};
  \end{scope}
\end{tikzpicture}
  \caption{Proof of the bound $D'({P_\ell})\le\ell-1$ in the case of $\ell=5$:
$a_1a_2a_3a_4a_5$ is a 5-path in $G$. Dashed segments are drawn between
non-adjacent vertices.}
  \label{fig:proof:P_n}
\end{figure}

In the first $\ell-2$ rounds, Spoiler pebbles the vertices $a_i$
for all $i$ except $i=1$ and $i=m$.
Suppose that Duplicator manages to respond in $H$ with two paths $b_2,\ldots,b_{m-1}$
and $b_{m+1},\ldots,b_\ell$.
Then Spoiler wins in one extra move by pebbling either $a_1$ or $a_m$.
Duplicator loses because in $H$ there are no two vertices $b_1$ and $b_m$
with the adjacency patterns to the already pebbled vertices that $a_1$ and $a_m$
have in $G$. Indeed, if such $b_1$ and $b_m$ exist,
they should be distinct because of different adjacency to $b_{m+1}$. But then
$b_1b_2\ldots b_\ell$ would be an $\ell$-path in $H$, a contradiction.

\Case 3{$m=2$.}
Note that $a_1$ is not adjacent to $a_3$.

\Subcase{3-a}{$a_1$ and $a_4$ are not adjacent.}
In the first $\ell-2$ rounds, Spoiler pebbles the vertices $a_i$
for all $i$ except $i=2$ and $i=4$.
Suppose that Duplicator manages to pebble vertices $b_1,b_3,b_5,\ldots,b_\ell$ in $H$
such that $b_5\ldots b_\ell$ is an $(\ell-4)$-path.
Then Spoiler wins in one extra move by pebbling either $a_2$ or $a_4$.
Duplicator loses because in $H$ there are no two vertices $b_2$ and $b_4$
with the same adjacency pattern to the already pebbled vertices.
Indeed, if such $b_2$ and $b_4$ exist,
they should be distinct because of different adjacency to $b_1$. This is impossible as
$b_1b_2\ldots b_\ell$ would be an $\ell$-path in $H$.

\Subcase{3-b}{$a_1$ and $a_4$ are adjacent; $a_2$ and $a_5$ are not.}
Spoiler wins in $\ell-1$ moves using exactly the same strategy
as in the preceding subcase (now $b_2$ and $b_4$ should be distinct
because of different adjacency to $b_5$).

\Subcase{3-c}{$a_1$ and $a_4$ are adjacent; $a_2$ and $a_5$ are also adjacent.}
Spoiler pebbles the $\ell-2$ vertices $a_3,a_4,\ldots,a_\ell$.
Suppose that Duplicator succeeds to respond with an $(\ell-2)$-path $b_3\ldots b_\ell$ in $H$.
Then Spoiler wins in the next round by pebbling either $a_1$ or $a_2$.
Duplicator loses because in $H$ there are no two vertices $b_1$ and $b_2$
with the same adjacency pattern to the already pebbled vertices.
Indeed, if such $b_1$ and $b_2$ exist,
they should be distinct because of different adjacency to $b_3$.
Note, however, that such $b_1$ and $b_2$ need not be adjacent.
In any case, $b_1b_4b_3b_2b_5\ldots b_\ell$ would be an $\ell$-path in $H$.

\medskip

\textit{The upper bound $W'({P_\ell})\le\ell-2$.}
Consider two connected graphs $G$ and $H$ both with at least $\ell$ vertices 
and such that $G$ contains a copy of $P_{\ell}$ while $H$ does not. 
Let us prove that Spoiler has a winning strategy in the $(\ell-2)$-pebble game in at most $\ell+1$ rounds.
This will mean that $G$ and $H$ are distinguishable by a first-order sentence $\Phi_{G,H}$ with
$\ell-2$ variables of quantifier depth $\ell+1$. Since there are only finitely many
pairwise inequivalent sentences of a bounded quantifier depth, the inequivalent sentences $\Phi_{G,H}$
for various $G$ and $H$ can be combined into a single $(\ell-2)$-variable sentence
defining the existence of a $P_{\ell}$ subgraph over all sufficiently large connected graphs.

Assume first that $H$ has the following two properties:
\begin{description}
\item[(A)] 
all pendant vertices have a common neighbor;
\item[(B)]
there are at most $\ell-2$ non-pendant vertices.
\end{description}
We show how Spoiler can win in this case by playing with $\ell-2$ pebbles and making at most $\ell+1$ moves.

If $G$ does not satisfy Condition (A), Spoiler pebbles two pendant vertices of $G$
with distinct neighbors. If Duplicator responds with two pendant vertices in $H$,
Spoiler wins by pebbling their common neighbor. If one of Duplicator's vertices
in $H$ is not pendant, Spoiler wins by pebbling two its neighbors.

Suppose now that $G$ satisfies (A). If $G$ has no pendant vertex, then Spoiler
pebbles a pendant vertex in $H$ and wins in the next two moves.
We, therefore, suppose that $G$ has at least one pendant vertex.
Denote the common neighbor of all pendant vertices in $G$ by $z$ and in $H$ by $z'$.
Condition (A) implies that $z$ is the central vertex of a pendant star in $G$.
Denote the number of non-pendant vertices in $G$ by $t$. Since $G$ contains $P_{\ell}$
as a subgraph, we have $t\ge\ell-1$. Let Spoiler pebble $\ell-2$ non-pendant vertices
in $G$ different from $z$. By Condition (B), Duplicator is forced to pebble either $z'$
or a pendant vertex in $H$. In the former case, Spoiler keeps the pebble on $z'$, 
moving one of the other pebbles to a pendant vertex in $H$.
Since the counterpart of $z'$ in $G$ is a non-pendant vertex different from $z$,
Duplicator is forced to respond with a non-pendant vertex in $G$.
In each case, Spoiler forces pebbling a pendant vertex in one graph
and a non-pendant vertex in the other graph, which allows him to
win in the next two moves.

Now, we assume that $H$ does not satisfy (A) or (B) and show that
then Spoiler is able to win with $\ell-2$ pebbles in at most $\ell$ moves.
Since $H$ contains no $P_{\ell}$, has at least $\ell$ vertices, and is connected,
it contains also no $C_{\ell-1}$. Let $x_1\ldots x_{\ell}$ be a copy of $P_{\ell}$ in $G$. 
In the first $\ell-2$ rounds, Spoiler pebbles the vertices $x_2,\ldots,x_{\ell-1}$.
Suppose that Duplicator does not lose immediately and responds with vertices $y_2,\ldots,y_{\ell-1}$ 
that form a path $P$ in $H$. In the next round, Spoiler chooses one of the vertices
$x_3,\ldots,x_{\ell-1}$ and moves the pebble from this vertex to $x_1$.
Duplicator can survive under any choice of Spoiler only if there is a neighbor of $y_2$ outside $P$ or 
when all $y_3,\ldots,y_{\ell-1}$ are adjacent to $y_2$. By a symmetric reason,
Duplicator does not lose in this round only if
there is a neighbor of $y_{\ell-1}$ outside $P$ or all $y_2,\ldots,y_{\ell-2}$ are
adjacent to $y_{\ell-1}$. 
It remains to show that these conditions imply Conditions (A) and~(B).

It is impossible that both $y_2$ and $y_{\ell-1}$ have neighbors outside $P$
because then $H$ would contain either $C_{\ell-1}$ or $P_{\ell}$ as a subgraph. 
Therefore, we can assume that at least one of the vertices $y_2$ and $y_{\ell-1}$ (say, $y_{\ell-1}$) 
is adjacent to all the other vertices in $P$. Consider two cases.

\Case 1 {$y_2$ has a neighbor $y$ outside $P$.}
It is easily seen that if $y$ were adjacent to at least one of $y_3,\ldots,y_{\ell-1}$, 
then $H$ would contain $C_{\ell-1}$. Moreover, if at least one of $y_3,\ldots,y_{\ell-1}$ had 
a neighbor outside $\{y,y_2,\ldots,y_{\ell-1}\}$, then $H$ would contain $P_{\ell}$. 
A $P_{\ell}$ subgraph would exist also if a neighbor of $y_2$ outside $P$ had another neighbor outside $P$.
It follows that, in this case, all the vertices outside $P$ are pendant in $H$ and are adjacent to $y_2$.
We conclude that $H$ satisfies Conditions (A) and (B) above.

\Case 2 {$y_2$ does not have any neighbor outside $P$.} 
In this case the vertex $y_2$, like the vertex $y_{\ell-1}$, is adjacent to all the other vertices of $P$. 
Any vertex outside $P$ must be pendant in $H$ because otherwise $H$ would contain $P_{\ell}$ or $C_{\ell-1}$.
Any two such vertices must have a common neighbor for else $H$ would contain $P_{\ell}$.
Again, $H$ must satisfy (A) and~(B).

\medskip

\textit{The lower bounds $D'({P_\ell})\ge\ell-1$ and $W'({P_\ell})\ge\ell-2$.}
The bound for the width parameter follows directly from Part 1 of Lemma \ref{lem:sp-lower}.
In order to prove the bound for the depth parameter, consider the graphs $G_n$ and $H_n$ 
as in the proof of this lemma. That is, $G_n$ consists of $K_{\ell-1}$ with a pendant
star $K_{1,n}$ and $H_n$ consists of $K_{\ell-2}$ also with a pendant star $K_{1,n}$,
where the parameter $n$ can be chosen arbitrarily large.
Obviously, $G_n$ contains $P_\ell$ as a subgraph, while $H_n$ does not.
The vertex set of each of the graphs $G_n$ and $H_n$ can be split into three classes:
the pendant vertices, the central vertex of the pendant $K_{1,n}$,
and the remaining non-pendant vertices.
Consider the Ehrenfeucht-Fra{\"\i}ss{\'e} game on $G_n$ and $H_n$.
Duplicator does not lose as long as she respects the vertex classes.
Spoiler can break this Duplicator's strategy only by pebbling
all vertices of $G_n$ in the third class, for which he needs no less than $\ell-2$ moves.
If he does this in the first $\ell-2$ rounds of the game, Duplicator still does
not lose if in the $(\ell-2)$-th round she responds with the vertex in the second class of $H$.
It follows that Spoiler cannot win by making only $\ell-2$ moves.
\end{proof}


\section{Further questions}\label{s:concl}
\mbox{}

Theorem \ref{thm:sparkler} shows that, for any sparkler graph, 
the values of $W'(S_{q,p})$ and $D'(S_{q,p})$ can differ from each other by at most 2.
The examples of path and star graphs show that $W'(F)$ and $D'(F)$ can have different values.
Specifically, we have $D'(F)=W'(F)+1$ by Theorems \ref{thm:K_1s-v} and \ref{thm:P_n}.
How large can be the gap between the two parameters in general?
Do we have $D'(F)\le W'(F)+O(1)$ for all~$F$?
If true, this would be a quite interesting feature of Subgraph Isomorphism
as, in general, first-order properties expressible with a bounded number
of variables can require an arbitrarily large quantifier depth.
Note in this respect that $D'(F)<\frac32\,W'(F)+3$ as a consequence
of Part 2 of Theorem \ref{thm:gen}, that is, the gap between $W'(F)$ and $D'(F)$
cannot be superlinear.

As it was already mentioned in Section \ref{s:intro},
the questions studied in the paper have perfect sense
also in the case that we want to express the existence
of an \emph{induced} subgraph isomorphic to the pattern graph $F$.
Let $D[F]$, $W[F]$, $D'[F]$, and $W'[F]$ be the analogs of
the parameters $D(F)$, $W(F)$, $D'(F)$, and $W'(F)$,
where the bracket parentheses emphasize that \emph{Induced}
Subgraph Isomorphism is considered.
The argument showing that $W(F)=\ell$ does not work any more in the
induced case and, indeed, $D[K_3+e]=3$ for the 4-vertex paw graph $K_3+e$.
We noticed this example in \cite{induced}, where we also proved a
general lower bound $W[F]\ge(\frac12-o(1))\ell$. If this bound
were tight for infinitely many $F$, Induced Subgraph Isomorphism
for these patterns would be solvable in time $O(n^{(\frac12-o(1))\ell})$
which would be interesting because the Ne\v{s}et\v{r}il-Poljak bound,
based on fast matrix multiplication, could never be better that
$O(n^{\frac23\,\ell})$. On the other hand, it is still not excluded
that $W[F]=\ell$ for all but finitely many $F$, which makes
the problem of determining or estimating the parameters $W[F]$ and $D[F]$
quite intriguing.

As for the asymptotic descriptive complexity over connected graphs,
we observed in \cite{induced} that $W'[F]=W[F]$ and $D'[F]=D[F]$
for a large class of pattern graphs, including all those without
universal vertices. In general, $W'[F]\ge(\frac13-o(1))\ell$ for all~$F$.


\end{document}